\pdfoutput=1
\documentclass[a4paper,11pt]{amsart}
\usepackage{Seidel}
\title[Seidel Minor]{Seidel Minor, Permutation Graphs \\and Combinatorial Properties
}
\author[V. Limouzy]{Vincent Limouzy}
\address{Limos - Univ. Blaise Pascal, Clermont-Ferrand, France}
\email{limouzy@isima.fr}

\begin{document}
\maketitle
\begin{abstract}
A permutation graph is an intersection graph of segments lying between 
two parallel lines. A Seidel complementation of a finite graph at 
a vertex $v$ consists in complementing the edges between 
the neighborhood  and the non-neighborhood  of $v$. Two graphs 
are \textit{Seidel complement} equivalent if one can be obtained 
from the other by a sequence  of \textit{Seidel complementations}.

In this paper we introduce the new concept of \textit{Seidel complementation} and 
\textit{Seidel minor}. We 
show that this operation preserves cographs and the structure of modular 
decomposition. 
\\The main contribution of this paper is  to provide 
a new and succinct characterization of permutation graphs 
namely, a graph is a permutation graph \Iff it does not contain any of  
 the following graphs: $C_5$, $C_7$, $XF_{6}^{2}$, $XF_{5}^{2n+3}$, $C_{2n}, n\geqslant6$
and their complements as a Seidel minor. 
This characterization is in a sense similar to Kuratowski's characterization 
\cite{Kuratowski30} of planar graphs by forbidden topological  minors. 
\\[12pt]
\end{abstract}
\section{Introduction}

A lot of graph classes  are frequently characterized by a list of forbidden 
induced subgraphs. For instance 
such characterization is known for  cographs, interval graphs, chordal graphs...
  However, it is not always convenient to deal with this 
kind of characterizations and  the list of forbidden subgraphs can be quite 
large. Some characterizations rely on the use of local operators such 
as minors, local complementation or Seidel switch. 

Certainly, Kuratowksi's characterization of planar graphs by forbidden 
topological minors is one of the most famous \cite{Kuratowski30}.

A nice characterization of circle graphs, \ie the intersection graphs  
 of chords in a circle, was given by Bouchet \cite{Bouchet94}, using 
an operation called \textit{local complementation}. This operation consists 
in complementing the graph induced by the neighborhood of a vertex. His 
characterization states that a graph is a circle graph \Iff  it does not contain 
 $W_5$,$W_7$ and $BW_3$\footnote{$W_5$ (\resp $W_7$) is the wheel on five  (\resp seven) vertices, 
\ie a chordless cycle  vertices plus a dominating vertex, and
$BW_{3}$ is a wheel on three vertices where the cycle is subdivided.} 
as vertex minor. This operation has strong connections with a graph 
decomposition called \textit{rank-width}, this relationship 
is presented in the work of Oum  \cite{Oum05,Oum05a}.

Another example of local operator is the Seidel switch.
The Seidel switch is a graph operator  introduced by Seidel in his seminal paper \cite{Seidel76}. 
 A Seidel switch in a graph consists in complementing the edges between a subset 
of vertices $S$ and its complement $V \setminus S$.  
 
Seidel switch 
 has been intensively studied since its introduction; 
Colbourn \etal \cite{ColbournC80} proved that  deciding whether two graphs 
are Seidel switch equivalent is 
\Iso-Complete. The Seidel switch has also applications 
in graph coloring \cite{Kratochvil03}. Other interesting applications 
of Seidel switch concerns structural graph properties \cite{Hayward96,Hertz99}.
It has also been used by Rotem and Urrutia \cite{RotemU82} to show 
that the recognition of circular permutation graphs (CPG for short) can be
 polynomially reduced to the recognition of permutation graphs. Years later, 
Sritharan  \cite{Sritharan96} presented a nice and efficient algorithm to 
recognize CPGs in linear time. Once again it is a reduction to permutation 
graph recognition, and it relies on the  use of a Seidel switch. 
 Montgolfier \etal \cite{MontgolfierR05,MontgolfierR05a} used it 
to characterize graphs completely decomposable  \wrt Bi-join decomposition. 
Seidel switch is not only relevant to the study of graphs.  
Ehrenfeucht \etal \cite{EhrenfeuchtHR99}
showed the interest of this operation for the study of $2$-structures and recently, 
Bui-Xuan \etal extended these results to broader structures called Homogeneous 
relations \cite{BuiXuanHLM07,BuiXuanHLM08,BuiXuanHLM08a}.

We present in this paper a novel characterization of the well known 
class of permutation graphs, \ie, the intersection graphs of segments lying between 
two parallel lines. 
Permutation graphs were introduced by Even, Lempel and Pnueli 
\cite{PnueliLE71,EvenPL72}. They established that a graph is  a permutation 
graph if and only if the graph and its complement are \textit{transitively orientable}.
They also gave a polynomial time procedure to find a transitive orientation 
when it is possible. A linear time algorithm recognition algorithm 
is presented in \cite{McConnellS99}.

This results 
constitutes, in a sense, an improvement  compared to Gallai's characterization of permutation 
graphs by forbidden induced subgraphs which counts no less than
18 finite graphs, and 14 infinite families \cite{Gallai67}.

For that we introduce a new local operator called Seidel complementation. 
In few words, the Seidel complementation on an undirected graph at a vertex $v$ consists 
in complementing the edges between the neighborhood and the non-neighborhood 
of $v$. A schema of Seidel complementation is depicted in 
Figure \ref{fig:SeidelComplement}.
Thanks to this operator and the corresponding minor, we obtain a 
compact list of Seidel minor obstructions for permutation graphs.

The main result of this paper is a new characterization of  permutation graphs.
We show that a graph is a permutation graph \Iff it does not contain any of
 the following graphs  $C_5$, $C_7$, $XF_{6}^{2}$, $XF_{5}^{2n+3}$, $C_{2n}, n\geqslant6$
or their complements as Seidel minors.

The proof is based on  a  study of the relationships between Seidel 
complementation and modular decomposition.
We show that any Seidel complementation of a prime graph \wrt modular decomposition 
is a prime graph. As a consequence we get that cographs are stable 
under Seidel complementation. We also present a complete characterization 
of equivalent cographs, which leads to a linear time algorithm for verifying 
Seidel complement equivalence of cographs.

Our notion of Seidel 
complementation is a combination of local complementation and Seidel switch. 
The use of a vertex  as \emph{pivot} comes from local 
complementation, and the transformation from Seidel switch.

The paper is organized as follows. In section \ref{sec:Intro} we present 
the definitions of Seidel complementation and  Seidel minor. 
Then  we show  some structural properties of 
Seidel complementation and we introduce the definitions and notations used
 in the sequel of the paper. In section \ref{sec:MDSeidel} we show  the 
relationships between Seidel complementation and modular decomposition, 
namely we prove that Seidel complementation preserves the structure of modular 
decomposition of a graph. Finally we show that cographs are closed under 
this relation. 
Section \ref{sec:PermutationGraph} is devoted to prove 
the main theorem, namely  a graph is a permutation graph 
\Iff it does not contain any of the forbidden Seidel minors. 
We also prove that permutation graphs are not  well quasi ordered under the  Seidel minor relation.
In section \ref{sec:equivalence} we show that any Seidel complement equivalent graphs are at distance                                                                                                                                      
at most one from each other, we show that to decide when two graphs are                                                                                                                                                                    
equivalent under the Seidel complement relation is \Iso-Complete and we provide                                                                                                                                                            
two polynomial algorithms to solve this problem on cographs and on permutation                                                                                                                                                             
graphs.                                                                                                                                                                                                                                    
And finally in section \ref{sec:tournaments} we propose a definition of Seidel complementation                                                                                                                                              
for tournaments, and we show, with that definition, we have same property                                                                                                                                                                  
\wrt modular decomposition as for undirected graphs.              

\section{Definitions and notations}
\label{sec:Intro}

In this paper only undirected, finite, loop-less and simple graphs are considered.
We present here  some notations used in the paper.
The graph induced by a subset of vertices $X$ is noted $G[X]$.
For a vertex $v$, $N(v)$ denotes the neighborhood of $v$, and $\noN{v}$ 
represents  the non-neighborhood.
Sometimes we need to use a refinement of the neighborhood on 
a subset of vertices $X$, noted $N_{X}(v) = N(v) \cap X$.
Let $A$ and $B$ be two disjoint subsets of $V$, and 
let $E[A,B] = \{ab \in E : a \in A \text{~and~} b \in B\}$ be the set of edges between 
$A$ and $B$. For two sets $A$ and $B$, let $A\Delta B = (A \setminus B) \cup (B\setminus A)$.

\begin{definition}[Seidel complement]
\label{def:SeidelSwitchN}
Let $G=(V,E)$ be a  graph, and let $v$ be 
a vertex of $V$, and the Seidel complement  at $v$ on $G$, denoted $G*v$ is defined as 
follows:
\\
Swap the edges and the non-edges between $G[N(v)]$ and $G[\noN{v}]$, namely 
\\
$$G * v  =(V, E \Delta \{xy : vx \in E, vy \notin E\})$$
\end{definition}

From the previous definition it is straightforward to notice that $G*v*v = G$.
\begin{proposition}
\label{prop:pivot}
Let $G$ be a graph. If  $vw$  is an edge of $G$,
then $G * v*w*v = G *w* v* w$. This operation is denoted $G \star vw$.
\end{proposition}
\begin{proof}
Let us consider the neighborhood of $v$ and $w$. Let $N(v) = A \cup B$ 
and let $N(w) = B \cup C$, where $B$ is obviously $N(v) \cap N(w)$,
and let $D$ be $\noN{v} \cap \noN{w}$.
We know that $v$ is connected to $A \cup B \cup \{w\}$ and 
that $w$ is connected to $B \cup C \cup \{v\}$. But we do not 
know how  the sets $A,~B,~C$ and $D$ are connected. We just 
say there are mixed edges between each set.
See Figure \ref{fig:WellDefined}.
\end{proof}

\begin{figure}[h!]
	\begin{center}
		\subfigure[G][]{\includegraphics[scale=.85]{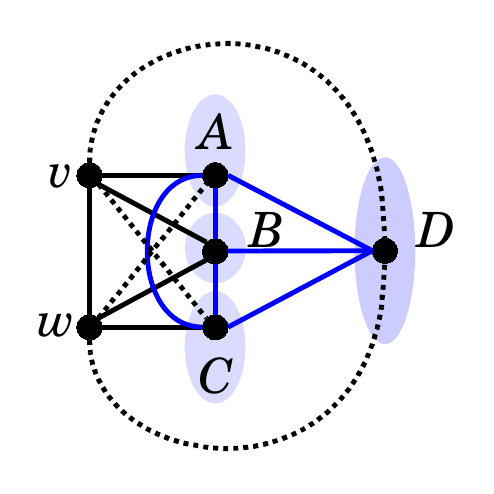}}
		\\
		\subfigure[][$G*v$]{
		\includegraphics[scale=.85]{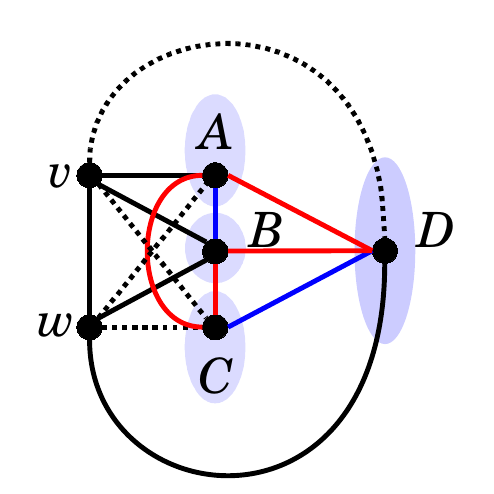}}
		\subfigure[][$G*v*w$]
		{\includegraphics[scale=.85]{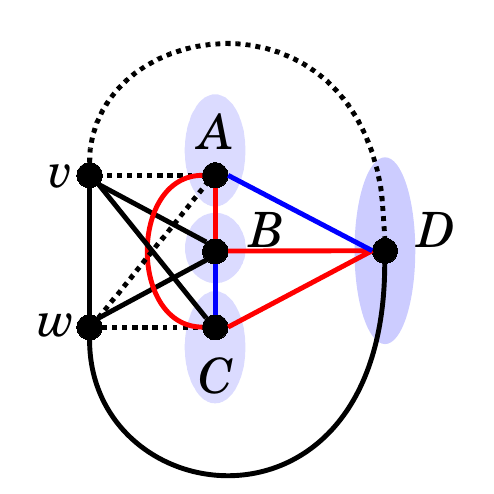}}
		\subfigure[][$G*v*w*v$]
		{\includegraphics[scale=.85]{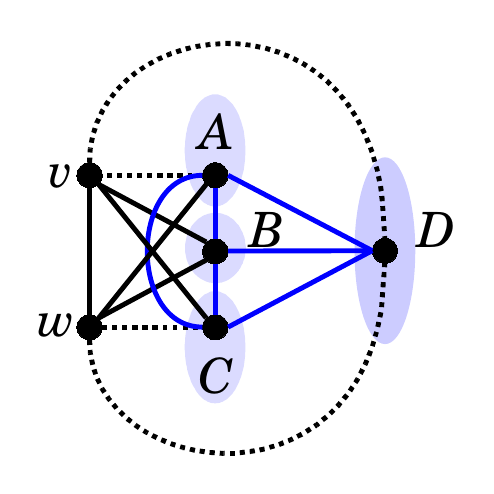}}
		\\
		\subfigure[][$G*w$]{\includegraphics[scale=.85]{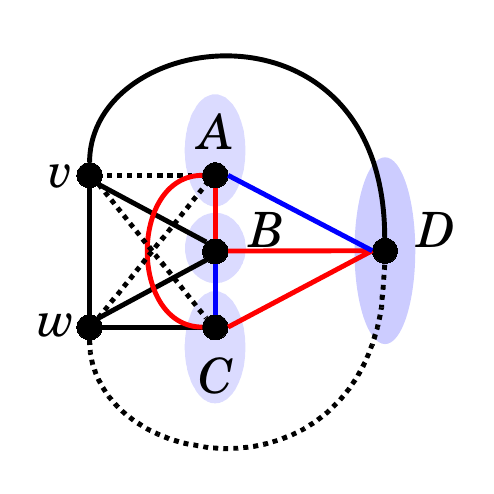}}
		\subfigure[][$G*w*v$]{\includegraphics[scale=.85]{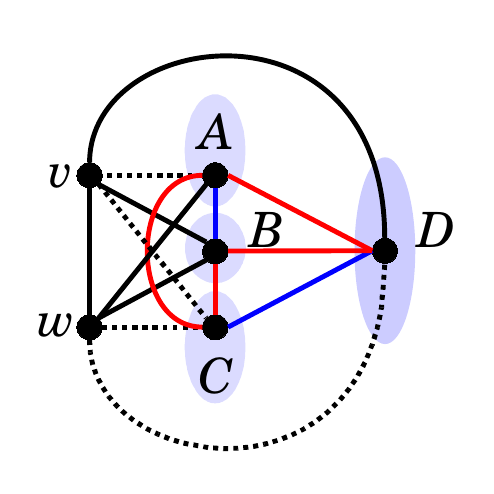}}
		\subfigure[][$G*w*v*w$]{\includegraphics[scale=.85]{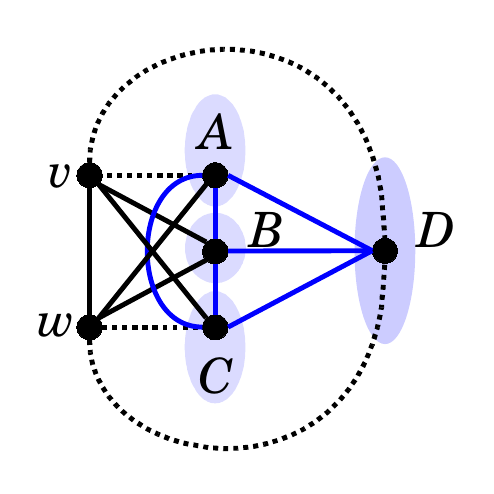}}
	\caption{$G*w*v*w = G*v*w*v = G \star vw$}
	\label{fig:WellDefined}
	\end{center}
\end{figure}
\begin{rem}[]
One can remark from Figure \ref{fig:WellDefined} that the $\star$ operation 
merely exchanges the vertices $v$ and $w$ without modifying the graph $G[V\setminus \{v\}\cup\{w\}]$.
\end{rem}

\begin{rem}[]
\label{rem:pivotnn}
 Proposition \ref{prop:pivot} remains true even if $v,w$ is not an edge of $G$. The proof 
 is similar to the proof of the Propositions \ref{prop:pivot}.
\end{rem}

\begin{definition}[Seidel Minor]
\label{def:SeidelMinor}
Let $G=(V,E)$ and $H=(V',E')$ be two graphs. $H$ is a Seidel minor of $G$
(noted $H \leqslant_{S} G$) if $H$ can be obtained from $G$ by a 
sequence of the following operations:
\begin{itemize}
	\item Perform a Seidel complementation at a vertex $v$ of $G$, 
	\item Delete a vertex  of $G$.
\end{itemize}
\end{definition}
\begin{definition}[Seidel Equivalent Graphs]
\label{def:SeidelEquivalent}
 Let $G=(V,E)$ and $H=(V,F)$ be two  graphs. $G$ and $H$ are said to be Seidel 
equivalent \Iff there exists a word $\omega$ defined on $V^{*}$ such 
that $G * \omega \cong H$.
\end{definition}
\begin{figure}[h!]
	\begin{center}
		\includegraphics[scale=1.5]{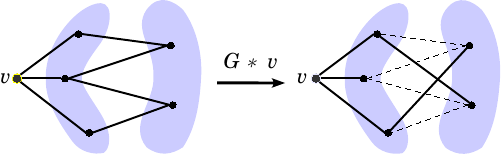}
	\caption{An illustration of the Seidel complement concept.}
	\label{fig:SeidelComplement}
	\end{center}
\end{figure}
At first glance Seidel complementation seems to be just a particular case                                                                                                                                                                  
of Seidel switch, but after a careful examination, one can see that they are                                                                                                                                                                     
not comparable.                                                                                                                                                                                                                            
                                                                                                                                                                                                                                            
Actually two graphs $G$ and $H$ that belong to the same Seidel switch                                                                                                                                                                     
equivalence class share a common combinatorial structure called a $2$-graph                                                                                                                                                                
\cite{Seidel76,EhrenfeuchtHR99}. A $2$-graph $\Omega=(V,D)$ is $3$-regular                                                                                                                                                                 
hypergraph where $V$ is the ground set and $D$ is the set of hyperedges, and                                                                                                                                                               
for each subset $S$ of $V$ of size $4$, we have $|D \cap S|\equiv 0~ mod~2$. A                                                                                                                                                                     
$2$-graph can be obtained from a graph by taking in $D$ all the triples of                                                                                                                                                                 
vertices with an odd number of edges.                                                                                                                                                                                                        
And from that definition we can see that the Seidel complementation applied                                                                                                                                                                
on a graph does not preserve the underlying $2$-graphs.                                                                                                                                                                                    
We can also see that with the Seidel complementation and with the Seidel                                                                                                                                                                   
switch starting from a same graph, the graphs we obtain with each operator                                                                                                                                                                 
are different. The reader can convince himself by looking at the house,                                                                                                                                                                    
\ie a cycle $C$ on five vertices plus a short chord connecting two vertices                                                                                                                                                                
at distance two in $C$.                

\section{Modular decomposition and cographs}
\label{sec:MDSeidel}
In this section we investigate the relationships between Seidel complementation 
and modular decomposition. This study is relevant in order to prove 
the main result. Actually a permutation graph is uniquely representable 
\Iff it is prime \wrt to modular decomposition. And one of the results 
of this section is to prove that if a graph is prime \wrt modular 
decomposition this property is preserved by Seidel complementation.
As a consequence for  permutation graphs, it means that 
if the graph is uniquely representable so are their Seidel complement 
equivalent graphs. 

Let us now briefly recall the definition of module.
A module in a graph is subset of vertices $M$  such that any vertex outside 
$M$ is either completely connected to $M$ or is completely disjoint from $M$. 
Modular decomposition is a decomposition of graph introduced by 
Gallai \cite{Gallai67}. The modular decomposition of a graph $G$  is the 
decomposition of $G$ into its modules. Without 
going too deeply into the details, there exists for each graph a unique modular 
decomposition tree, and it is possible to compute it in linear time (\cf 
\cite{TedderCHP08}).

In the sequel of this section we show that if $G$ is prime, \ie not 
decomposable, \wrt modular decomposition, then applying a Seidel complementation 
at any vertex of the graph preserves this property. 
Then we prove that the family of cographs is closed under Seidel minor. 
And finally show how the modular decomposition tree of a graph 
 is modified by a Seidel complementation.
\subsection{Modular decomposition}
\begin{theorem}[]
\label{th:SeidelMDPrime}
Let $G=(V,E)$  be graph, and let $v$ be an arbitrary vertex of $G$. 
$G$ is  prime \wrt modular decomposition \Iff $G*v$ is prime \wrt to modular 
decomposition.
\end{theorem}
\begin{proof}
Let us proceed by contradiction. Let us assume that $G$ is prime and 
$G * v$ has a module $M$. 

We have to consider two cases: 
(1) $v \in M$ and 
(2) $v \notin M$.
\\
\textbf{(1)}
Since $M$ is not trivial: we have $|M| \geqslant 2$ and 
$|\overline{M}|\geqslant 1$.
\\
We can identify four representative vertices of $G$: let $A$ 
be a vertex of $\noN{v}\cap M$, let $B$ be a vertex of 
$\noN{v} \cap \overline{M}$, let $C$ be a vertex of 
$N(v) \cap \overline{M}$ 
and let $D$ be a vertex of $N(v) \cap M$. Since $M$ is a module 
we have the following edges: $CA$ and $CD$ and the following non-edges:
$BA$ and $BD$ (\cf Figure \ref{fig:VinM-a}). 

By definition of Seidel complementation at  a vertex, it is equivalent 
to swap the edges and non-edges between the neighborhood and 
the non-neighborhood of $v$. We obtain the result 
depicted in Figure \ref{fig:VinM-b}. Now we can clearly see 
that $\overline{M} \cup \{v\}$ is a module in $G$, and 
since $|\overline{M}|\geqslant 1$ we obtain a non-trivial module. 
Thus a contradiction.
\\
\noindent
\textbf{(2)}
Let us consider the case where $v$ does not belong to $M$. 
We can assume, \Wlog, that $M \subseteq N(v)$. We can partition 
$N(v)$ into $A_1,A_2$ such that $N_M(A_1) = M$ and $N_M(A_2) = \varnothing$. 
And similarly we can partition $\noN{v}$ into $B_1,B_2$ such 
that $N_M(B_1) = M$ and $N_M(B_2)= \varnothing$. 
(\cf Figure \ref{fig:vNotinM-a}-\subref{fig:vNotinM-b})

Since we have proceeded to a Seidel complement on $v$, the original configuration 
in $G$ is such that $N_M(B_1) = \varnothing$ and $N_M(B_2) = M$. This is the 
only modification \wrt $M$. Thus $M$ is also a module in $G$. Contradiction.
\end{proof}
%
From the second case of the proof of Theorem \ref{th:SeidelMDPrime} we can deduce 
the following corollary:

\begin{corollary}
\label{cor:BackboneModule}
Let $G=(V,E)$ and let $v$ be a vertex. And let $M$ be a module 
of $G$ such that $v$ does not belong to $M$, then $M$ is also a module 
in $G * v$.
\end{corollary}

%
\begin{figure}[h!]
	\begin{center}
		\subfigure[][$G*v$]
		{\label{fig:VinM-a}
		\includegraphics[scale=.80]{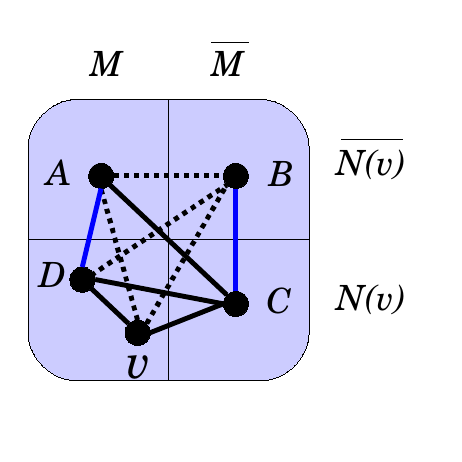}}
		\subfigure[][$G$]
		{\label{fig:VinM-b}
		\includegraphics[scale=.80]{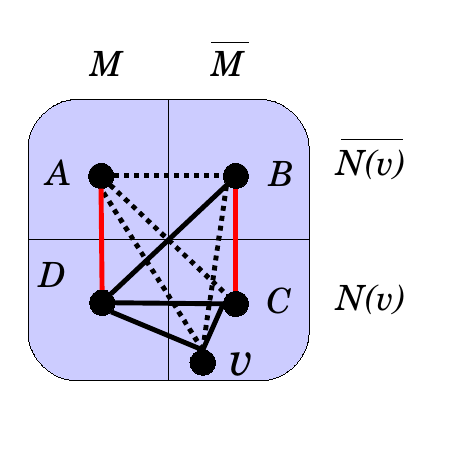}}
\hfill
		\subfigure[][Configuration in $G * v$ \label{fig:vNotinM-a}]
		{\includegraphics[scale=.80]{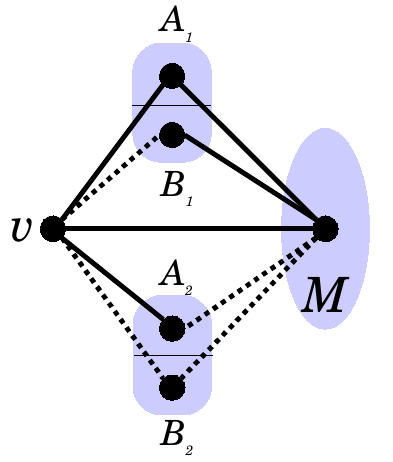}}
		\subfigure[][Configuration in $G$ \label{fig:vNotinM-b}]
		{\includegraphics[scale=.80]{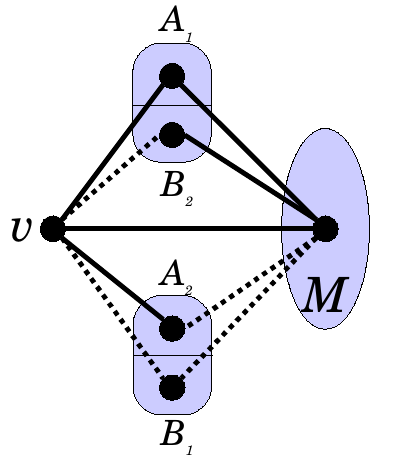}}
	\caption{Details of theorem \ref{th:SeidelMDPrime}. The Figures 
\subref{fig:VinM-a}-\subref{fig:VinM-b} correspond to the case where $v$ belongs 
to $M$. And the Figures \subref{fig:vNotinM-a}-\subref{fig:vNotinM-b} correspond 
to the other case.}
	\label{fig:VinM}
	\label{fig:vNotinM}
	\end{center}
\end{figure}
%
\subsection{Cographs}

Cographs are the graphs which are completely decomposable \wrt modular decomposition. 
There exist several characterizations of cographs (see \cite{CorneilPS85}), one of them  
is given by a forbidden induced subgraph, \ie cographs are the graphs without $P_4$ --a 
chordless  path on four vertices-- as induced subgraph. Another fundamental property 
of cograph is the fact that its modular decomposition tree --called 
its co-tree-- has only series (1) and parallel (0) nodes as internal nodes. 
An example of a cograph and its associated co-tree is given in Figure \ref{fig:PetitExemple}.
A co-tree is a rooted tree, where the leaves represent the vertices of the graph, 
and the internal nodes of the co-tree encode the adjacency of the vertices of the graph.
Two vertices are adjacent iff their Least Common Ancestor\footnote{The \LCA of two leaves 
$x$ and $y$ is first node in common on the paths from the leaves to the root.} 
(\LCA) is a series node (1). 
Conversely two vertices are disconnected iff their \LCA is a parallel node (0).
The following theorem shows that the class of cographs is closed under Seidel 
complementation.
\begin{theorem}[]
\label{th:SeidelCograph}
Let $G=(V,E)$ be a cograph, and $v$ a vertex of $G$, then 
$G * v$ is also a cograph.
\end{theorem}
\begin{proof}
Let $T$ be the co-tree of $G$. The Seidel complementation at a vertex $v$ 
is obtained as follows: 
Let $T'$ be the tree obtained by $T*v$.
$P(v)$, the former parent node of $v$,  becomes the new root of $T'$, 
and now the parent of $v$ in $T'$ is the former root, namely $R(T)$. In other words 
by performing a Seidel complementation we have reversed the path from
 $P(v)$ to $R(T)$. 
\\
It is easy to see that $G[N(v)]$ and $G[\noN{v}]$ are not modified.
Now to see that  the adjacency between $G[N(v)]$ and $G[\noN{v}]$ 
is reversed, it is sufficient to remark that for two vertices $u$ and $w$, $u$ belonging 
to the neighborhood of $v$ and $w$ belonging to the non-neighborhood of $v$. 
If $u$ and $w$  are adjacent in $G$ 
it means that their \LCA is a series node. 
We note that this node lies on the path from $v$ to $R(T)$.
After proceeding to a 
Seidel complementation their \LCA is modified and it is now a parallel 
node, consequently reversing the adjacency between the neighborhood and 
the non-neighborhood.
\end{proof}
An example of the Seidel complement of the co-tree is given in Figure 
\ref{fig:ExCographeSeidelComp}.
\begin{rem}[Exchange property]
\label{rem:exchange}
Actually a Seidel complementation on a cograph, or more precisely on 
its co-tree is equivalent of exchanging the root of the co-tree with 
the vertex $v$ used to proceed to the Seidel complement, \ie 
the vertex $v$ is attached to the former root of the co-tree and 
the new root is the former parent of the vertex $v$. 

Except for this transformation, the other parts of the co-tree remain 
unchanged, \ie the number and the types of internal nodes 
are preserved, and no internal nodes are merged. 
\end{rem}
\begin{figure}[h!]
	\begin{center}
		\subfigure[][]
		{\label{fig:PetitExemple}
		\includegraphics[scale=.80]{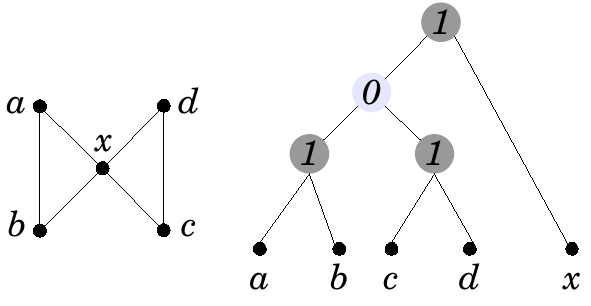}}
\hfill
		\subfigure[][]
		{\label{fig:ExCographeSeidelComp}
		\includegraphics[scale=1.50]{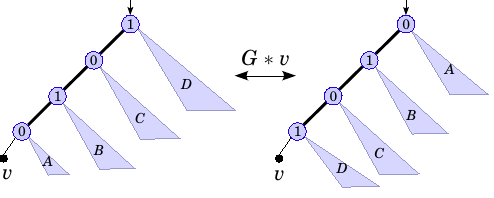}}
	\caption{\subref{fig:PetitExemple} An example of a cograph on 5 vertices and its 
respective co-tree. \subref{fig:ExCographeSeidelComp} A schema of a Seidel 
	complement at a vertex $v$ on a co-tree.}
	\label{fig:2}
	\end{center}
\end{figure}
\begin{proposition}[]
The Seidel complementation of a cograph on its co-tree 
can be performed in $O(1)$-time.
\end{proposition}
\begin{proof}
It suffices to consider the co-tree of $G$. As noticed in remark \ref{rem:exchange} 
to perform a Seidel complementation at a vertex $v$ is equivalent to 
exchange a vertex -- \ie a leaf -- with the root of the tree.  
We  need to store, in a lookup table, for each vertex its parent node in 
the tree   and the root of the tree. Updating  the structure can easily be 
done in constant time.
\end{proof}
%
\subsection{Modular decomposition tree}
In this section we will show how the modular decomposition tree of 
a graphs is modified. Using Theorems \ref{th:SeidelMDPrime},  
\ref{th:SeidelCograph} and \ref{cor:BackboneModule}

Let $G=(V,E)$ be a graph, and let $T(G)$ ($T$ for short) be its modular 
decomposition tree. Modular decomposition tree is a generalization of the co-tree 
for cographs. The only difference with co-tree is that the modular decomposition 
tree can contain prime nodes. Prime nodes corresponds to graphs that are 
not decomposable \wrt modular decomposition.

We generalize the operation on the co-tree, described in Theorem \ref{th:SeidelCograph},                                                                                                                                                   
 to arbitrary modular decomposition tree.

\begin{theorem}
\label{th:MDtreeSeidel}
Let $G=(V,E)$ be a graph, and let $T$ be its modular decomposition tree. 
Let $v$ be a vertex of $G$. By applying a Seidel complement at $v$ 
the modular decomposition tree of $T*v$ of $G * v$ is obtained 
by: 
\begin{itemize}
  \item performing a Seidel complement in every prime node  lying on the path 
from $v$ to $R(T)$.
  \item making $P(v)$ the root of $T*v$. 
  \item Reverse the path from $P(v)$ to $R(T)$: 
 if $\alpha$ and $\beta$ are prime node in $T$ with $\beta=P(\alpha)$ 
then $\alpha = P(\beta)$ and $\beta$ is connected in place of the subtree 
coming from $v$.
\end{itemize}
\end{theorem}

\begin{proof}
If $G$ is prime \wrt modular decomposition, its modular decomposition tree has 
only one internal node labeled prime, and the leaves represent the vertices. 
And we know by theorem \ref{th:SeidelMDPrime} that $G*v$ is also prime. 

When  the graph is not prime and is not a cograph, it admits a modular 
decomposition tree with more than one internal node. We have seen in Theorem 
\ref{th:SeidelCograph} that when all the internal nodes are of type parallel 
or series, the statement holds.

It remains to deal with the case of prime nodes. We have to notice, as a 
consequence of Corollary \ref{cor:BackboneModule}, that any module that do not 
contain $v$ are not impacted by the Seidel complement at $v$. It means that 
only the modules that contain $v$  are modified by the Seidel complement. 
And all the module that contain $v$ are precisely the nodes lying on the path 
from $v$ to $R(T)$.

The next part of the proof is illustrated in Figure \ref{fig:MDTAlphaBeta}. \\
Let us consider the case when the path from $v$ to the root of $T$ is 
constituted of two prime nodes $\alpha$ and $\beta$, with $\beta=P(\alpha)$ and 
$\alpha$ is connected to $\beta$ on the vertex $b$ and $v$ is connected to $\alpha$ 
on the vertex $a$.
Since $\alpha$ is a module, all the vertices connected to $\alpha$ are also connected 
to $N_\beta(b)$ the neighbors of $b$ in $\beta$. 
By performing a Seidel complementation at $v$ we must remove the edges 
between $N(v)$ and $\overline{N(v)}$, in particular 
we must disconnect $\overline{N_\alpha(a)}$ from $N_\beta(b)$  and 
we must connect $\overline{N_\beta(b))}$ to $N_\alpha(a)$. 

By performing the Seidel complement in $\alpha$ (resp. $\beta$) at $a$ (resp. $b$), 
we  satisfy the condition of Seidel complementation in each prime node. 
Now we need to realize the conditions above mentioned. 
By making now $\alpha$ the root and by connecting $\beta$ to $a$ and 
connecting $v$ to $b$. The condition is realized. Now since $\beta$ is 
a module attached under $\alpha$ every vertex contained in $\beta$ 
is connected to the the neighbors in $\alpha$. And every vertex 
connected of $\overline{N_\alpha(a)}$ are no longer connected to 
$N_\beta(b)$. We still have $N_\alpha(a)$ completely connected to 
$N_\beta(b)$ and $\overline{N_\alpha(a)}$ disconnected from $\overline{N_\beta(b)}$.

We can easily generalize this  to paths of height greater  than 2.
\end{proof}

\begin{figure}[h!]
	\begin{center}
		\includegraphics[scale=1.5]{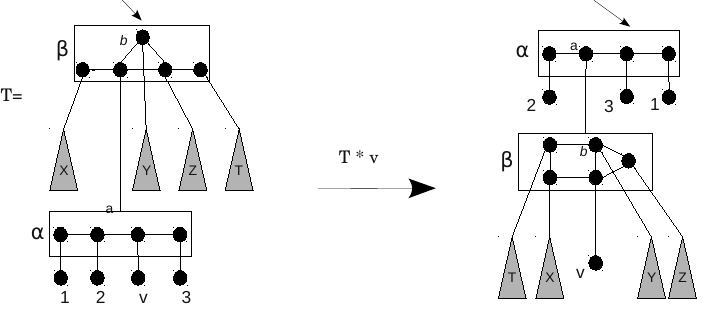}
	\caption{Effects of Seidel complementation on a modular decomposition tree.}
	\label{fig:MDTAlphaBeta}
	\end{center}
\end{figure}
\section{Permutation graphs}
\label{sec:PermutationGraph}
In this section we show that the class of permutation graphs is closed under 
Seidel minor, and we prove the main theorem that states 
that a graph is a permutation graph \Iff it does not 
contain any of the following graphs: $C_5$, $C_7$, $XF_{6}^{2}$, 
$XF_{5}^{2n+3}$, $C_{2n}, n\geqslant6$ or their complements as Seidel minor.

\begin{definition}[Permutation graph]
\label{def:PermutationGraph}
A graph $G=(V,E)$ is a permutation graph if there exist 
two permutations $\sigma_{1},\sigma_{2}$ on $V=\{1,\ldots,n\}$, 
such that two vertices $u,v$ of $V$ are adjacent iff $\sigma_{1}(u) < \sigma_{1}(v)$ 
and  $\sigma_{2}(v) < \sigma_{2}(u)$. 
$R = \{\sigma_1,\sigma_2\}$ is called a representation of $G$, and $G(R)$ is the permutation 
graph represented by $R$.
\end{definition}

More properties of permutation can be found in \cite{Golumbic04}.
An example of a permutation graph is presented in Figure \ref{fig:PermGraph}.
\begin{figure}[h!]
	\begin{center}
		\includegraphics[scale=1.0]{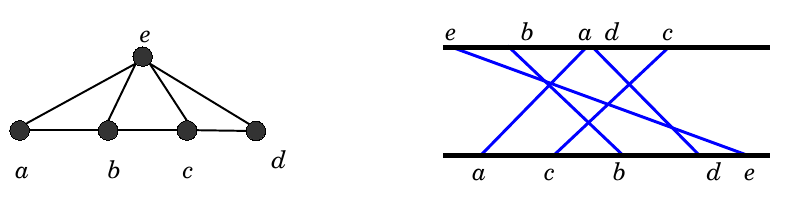}
	\caption{A permutation graph and its representation.}
	\label{fig:PermGraph}
	\end{center}
\end{figure}
\begin{theorem}[Gallai'67 \cite{Gallai67}]
A permutation graph is uniquely representable iff it is prime 
\wrt modular decomposition.
\end{theorem}
\begin{theorem}
[\cite{Gallai67}\footnote{\url{http://wwwteo.informatik.uni-rostock.de/isgci/classes/AUTO_3080.html}}]
\label{th:GallaiPermutation}A graph is a permutation graph \Iff it does not 
contain one of the finite graphs as induced subgraphs 
$T_2$, $X_2$, $X_3$, $X_{30}$, $ X_{31}$, $X_{32}$, $X_{33}$, $X_{34}$, $X_{36}$ 
 nor their complements  and does not contain the graphs 
given by the infinite families:
$XF_{1}^{2n+3}$, $XF_{5}^{2n+3}$, $XF_{6}^{2n+2}$, $XF_{2}^{n+1}$, $XF_{3}^{n}$, $XF_{4}^{n}$, 
the Holes, and their complements.
\end{theorem}
\begin{figure}[h!]
	\begin{center}
		\framebox{\subfigure[][\label{fig:FiniteFPerm-T2}$T_2$]
		{\includegraphics[scale=1.0]{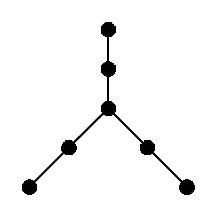}}}
		\framebox{\subfigure[][\label{fig:FiniteFPerm-X31}$X_{31}$]
		{\includegraphics[scale=1.0]{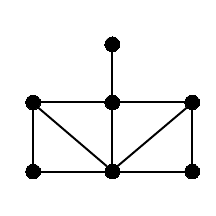}}}
		\framebox{
		\subfigure[][$X_2$]{\label{fig:X2}\includegraphics[scale=1.0]{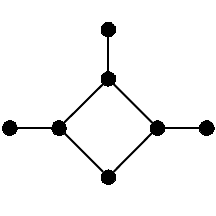}}
		\subfigure[][$X_3$]{\label{fig:X3}\includegraphics[scale=1.0]{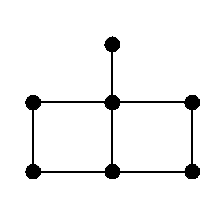}}
		\subfigure[][$X_{36}$]{\label{fig:X36}\includegraphics[scale=1.0]{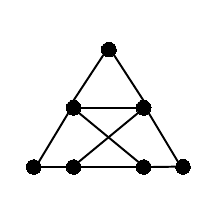}}}
\\[9pt]		
		\framebox{
		\subfigure[][$X_{30}$]
		{\label{fig:X30}\includegraphics[scale=1.0]{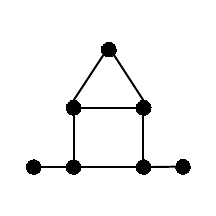}}
		\subfigure[][$X_{32}$]
		{\label{fig:X32}\includegraphics[scale=1.0]{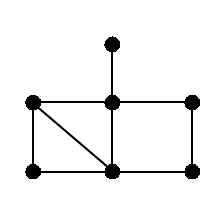}}
		\subfigure[][$X_{33}$]
		{\label{fig:X33}\includegraphics[scale=1.0]{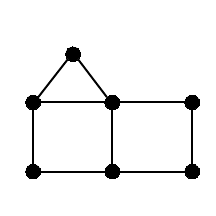}}
		\subfigure[][$X_{34}$]
		{\label{fig:X34}\includegraphics[scale=1.0]{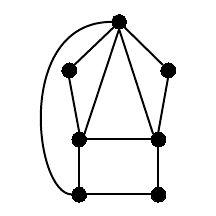}}
		}
	\caption{Finite forbidden induced subgraphs for permutation graphs.
	Those in the same box are Seidel complement equivalent.}
	\label{fig:FiniteFPerm}
	\end{center}
\end{figure}
\begin{figure}[h!]
	\begin{center}
		\framebox{
		\subfigure[][Hole]{\label{fig:Hole}\includegraphics[scale=1.0]{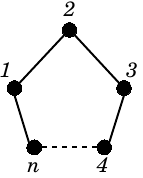}}
		\subfigure[][$XF_{2}^{n+1}$~]
		{\label{fig:XF2}\includegraphics[scale=1.0]{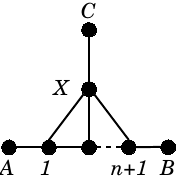}}
		\subfigure[][$XF_{3}^{n}$]
		{\label{fig:XF3}\includegraphics[scale=1.0]{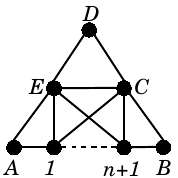}}
		\subfigure[][\label{fig:InfinitePerm-XF4}$XF_{4}^{n}$]
		{\includegraphics[scale=1.0]{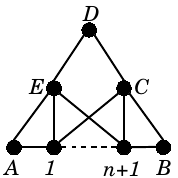}}
		}
		\framebox{
		\subfigure[][$XF_1^{n}$]{\label{fig:XF1}\includegraphics[scale=1.0]{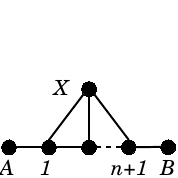}}
		\subfigure[][$XF_{5}^{n}$]{\label{fig:XF5}\includegraphics[scale=1.0]{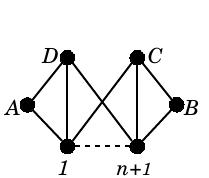}}
		\subfigure[][$XF_{6}^{n}$]{\label{fig:XF6}\includegraphics[scale=1.0]{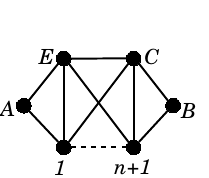}}
		}
	\caption{Forbidden infinite families for permutation graphs.
	The families in the left box \subref{fig:Hole}-\subref{fig:InfinitePerm-XF4} 
	contains asteroidal triples. The families in the right box 
\subref{fig:XF1}-\subref{fig:XF6} do not 
contain asteroidal triple, the key point is the parity of the dashed path.}
	\label{fig:InfinitePerm}
	\end{center}
\end{figure}
\noindent
{\bf Operation S:}
Let $\sigma = A~.~v~.~B$  be a permutations on $[n]$. Let $v$ be 
an element of $[n]$. The operation $S$ at an element $v$ of $[n]$ noted $\sigma * v$ is done 
Let $\sigma * v = B ~ .~ v~ .~ A$.

\begin{rem}[]
Let $R=\{\sigma_1,\sigma_2\}$ be a permutation representation of a permutation graph $G$ 
then $R * v = \{\sigma_1*v,\sigma_2*v\}$ is a permutation representation of a graph $H$. 
\end{rem}

\begin{theorem}[]
\label{th:SeidelPermutation}
Let $G=(V,E)$  be a permutation graph, and  let $v$ be a vertex of $G$, and 
let $R=\{\sigma_1,\sigma_2\}$ be the permutation representation of $G$. 
We have $G(R * v) = G * v$.
\end{theorem}
\begin{proof}
Let $G=(V,E)$ be  a permutation graph and $v$ a vertex of $G$. Let us prove 
that $G*v$ remains a permutation graph.                                                                                                                                                                                                    
Operation $S$ applied simultaneously on $R=\{\sigma_1,\sigma_2\}$ 
 is depicted in Figure \ref{fig:PermScheam-A}.
Let $\sigma_{1}$ be $A ~ .~ v~ .~ B$ and   $\sigma_{2}$ be $C ~ .~ v~ .~ D$, 
\\
The operation $S$ on $R=\{\sigma_1,\sigma_2\}$                                                                                                                                                                                             
corresponds to a Seidel complementation at $v$.
We have to prove that the graphs induced by the neighborhood $G[N(v)]$ and 
$G[\noN{v}]$ are unchanged. Let us begin with the non-neighborhood of 
$v$. It is easy to notice on Figure \ref{fig:PermScheam-A} that the 
non-neighborhood of $v$ is contained in the two vertical rectangles, 
one on the left of $v$ and the other one on their right, $(A,C)$ and 
$(B,D)$. By proceeding to the transformation described above, and by 
keeping the order of the words, it is easy to notice that 
first of all, these vertices remain disconnected from $v$ and since 
the order of vertices in the words are preserved then this subgraph 
remain unchanged. In a similar manner for the subgraph induced by 
the neighborhood of $v$, now the vertices of their neighborhood 
are contained in the gray crosses $(A,D)$ and $(B,C)$ and 
for the same reason as for the non-neighborhood, the subgraph 
remains unchanged and it is still connected to $v$.
\\
Now let us consider the less obvious part which is to swap 
the adjacency between $G[N(v)]$ and $G[\noN{v}]$.
Let $w$ be a neighbor of $v$ and let $u$ be a non-neighbor of $v$.
Let us assume, \Wlog, that $w$ and $u$ are connected. Let us consider 
the case where $u$ belongs to the $(A,C)$ rectangle and $w \in (A,D)$, if $uw \in E$ 
it means that $\sigma_{1}(w)< \sigma_{1}(u)$ and $\sigma_{2}(u)< \sigma_{2}(w)$, 
after proceeding to a Seidel complement at $v$ we obtain $\sigma'_{1} = \sigma_{1} * v$ 
and $\sigma'_{2 }=\sigma_{2} * v$ but now according to the transformation we have 
$\sigma'_{1}(w)< \sigma'_{1}(u)$ and $\sigma'_{2}(w)< \sigma'_{2}(u)$. 
And according to the definition \ref{def:PermutationGraph} now 
$u$ and $w$ are no longer connected. The proof is similar for the other 
cases. 
\end{proof}
\begin{figure}[h!]
	\begin{center}
		\subfigure[][\label{fig:PermScheam-A}]
		{\includegraphics[scale=1]{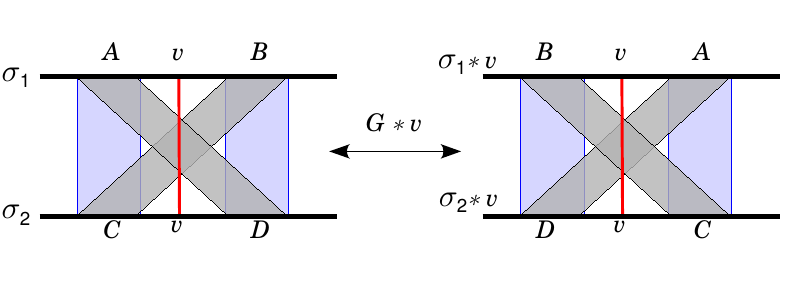}}
		\hfill
		\subfigure[][\label{fig:PermScheam-B}]
		{\includegraphics[scale=.80]{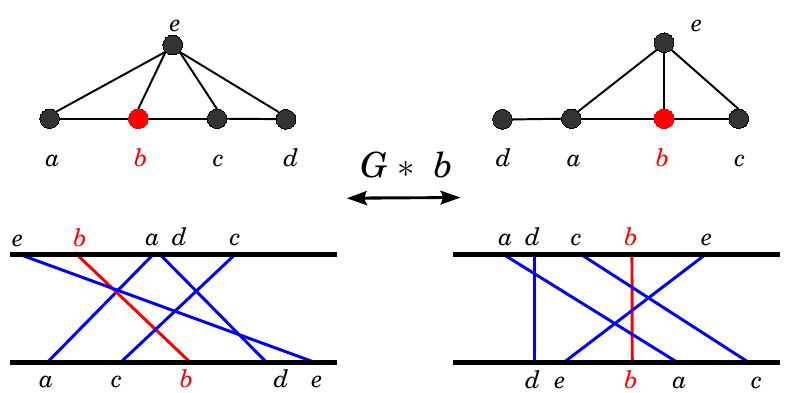}}
	\caption{\subref{fig:PermScheam-A} Schematic view of Seidel 
complementation on a permutation diagram. 
	 \subref{fig:PermScheam-B} An example of a permutation graph and a Seidel complementation 
		at a vertex $b$.}
	\label{fig:PermSchema}
	\end{center}
\end{figure}
\begin{corollary}[]
\label{cor:SeidelPermConstantTime} 
The Seidel complementation at a vertex $v$ of a permutation graph 
can be achieved in $O(1)$-time.
\end{corollary}
\begin{proof}
It is sufficient to consider the permutation representation of $G$ 
as two doubly linked lists. Then the Seidel complementation consists 
of applying the pattern described in the proof of Theorem \ref{th:SeidelPermutation}.
It consists \Wlog on $\sigma_1$ to exchange $A$ and $B$: 
$A \cdot v \cdot B$ becomes $B \cdot v \cdot A$. So it suffices 
to change the successor of $v$ in the list as the first element of $A$ 
and the predecessor of $v$ as the last element of $B$. Then update 
the first and last element of the new list. We 
proceed similarly for $\sigma_{2}$. All these operations 
can obviously be done in constant time.
\end{proof}

An arbitrary remark. To perform a Seidel complementation at a vertex on a graph 
can require in the worst case $O(n^{2})$-time. It suffices to 
consider  the graph consisting  of a star $K_{1,n}$ and a stable $S_{n}$,
whose size is $2n+1$ with $n+1$ connected components. 
Applying a Seidel complementation 
on the vertex of degree $n$ results in a connected graph with $O(n^{2})$ edges. 
\subsection{Finite Families}
In this section we show that it is possible to reduce the 
list of forbidden induced subgraphs by using Seidel Complementation. 
Actually a lot of forbidden subgraphs are Seidel equivalent. 
The graphs that are Seidel complement equivalent are in the same 
box in Figure \ref{fig:FiniteFPerm}. 
Thus, the list of finite forbidden graphs is reduced from 18 induced subgraphs 
to only 6 finite Seidel minors. The forbidden Seidel minors are 
\FListc and their complements.
\begin{proposition}[]
\label{prop:X3X2X36}
The graphs $X_3$, $X_2$, $X_{36}$ (\cf Figure \ref{fig:X2}-\subref{fig:X36})
 are Seidel complement equivalent.
\end{proposition}
\begin{proposition}[]
\label{prop:X30X32X33X34}
The graphs $X_{30}$, $X_{32}$, $X_{33}$ and $X_{34}$ 
(\cf Figure \ref{fig:X30}-\subref{fig:X34}) are Seidel complement equivalent.
\end{proposition}
Proofs of propositions \ref{prop:X3X2X36} and \ref{prop:X30X32X33X34} are 
not presented here, they essentially consist for each graph to find 
which vertex allows us to transform one graph into another.

The following proposition show that two forbidden finite graphs 
contain actually an instance of a member of an infinite family as 
Seidel minor. Thus it is no longer necessary to keep them 
in the list of forbidden Seidel minors. 
\begin{proposition}[]
\label{prop:XF40-SM-T2X31}
The graph $XF_{4}^{0}$ is a Seidel minor of $T_2$ and $X_{31}$.
\end{proposition}
\begin{proof}
$XF_{4}^0 <_{S} T_{2}$ 
\begin{center}
	\includegraphics[scale=1]{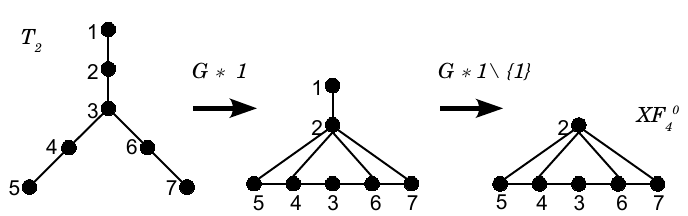}
\end{center}
and $XF_{4}^0 <_{S} X_{31}$
\begin{center}
	\includegraphics[scale=1]{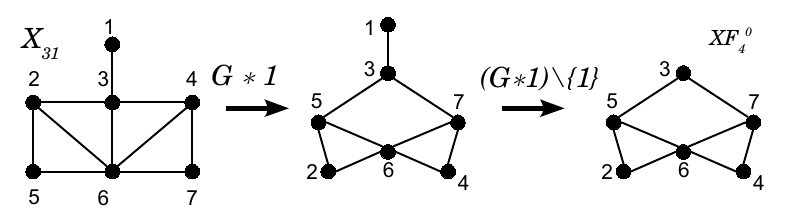}
\end{center}
\end{proof}
\begin{proposition}[]
\label{prop:C6XF40}
The graph $C_6$ is a Seidel minor of $XF_{4}^{0}$.
\end{proposition}
\begin{proof}
Applying a Seidel complementation on the degree 2 vertex of 
the $C_4$ in $XF_{4}^{0}$ we obtain $C_6$.
\end{proof}
\subsection{Infinite Families}
We show in this section that actually forbidden infinite families under 
the relation on induced subgraphs are redundant when the Seidel minor 
operation is considered. Consequently the following propositions 
allows us to reduce from 14 infinite families with the induced subgraph 
relation to only 4 infinite families under Seidel minor relation. 
The forbidden families are \IFList and their complements.
\begin{proposition}[]
\label{prop:Hole-SM-XF3XF4XF2}
The $Hole$  is a Seidel minor of $XF_{3}^{n}$, $XF_{4}^{n}$  and $XF_{2}^{n+1}$.
\end{proposition}
\begin{proof}
\begin{center}
\includegraphics[scale=1.0]{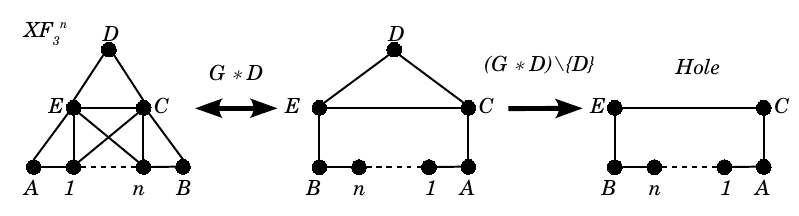}
\end{center}
\begin{center}
\includegraphics[scale=1.0]{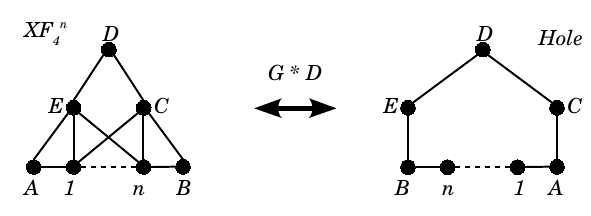}
\end{center}
\begin{center}
	\includegraphics[scale=1]{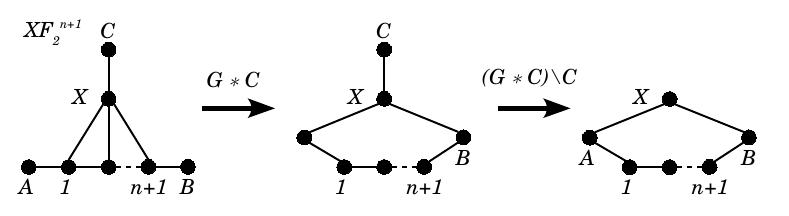}
\end{center}
\end{proof}
\begin{proposition}[]
\label{prop:XF6n}
$XF_{5}^{2n+1}$ is a Seidel minor of $XF_{6}^{2n+2}$.
\end{proposition}
\begin{proposition}[]
\label{prop:XF1n}
$XF_{5}^{2n+1}$ is a Seidel minor of $XF_{1}^{2n+3}$.
\end{proposition}
\begin{proposition}[]
\label{prop:C2n+1}
$XF_{5}^{2n+1}$ is a Seidel minor of $C_{2n+3}$.
\end{proposition}
\subsection{Main Theorem}
\begin{definition}[Seidel Complement Stable]
\label{def:SeidelStable}
A graph $G=(V,E)$ is said to be Seidel complement  stable if:
$ \forall v \in V: G \cong G*v$
\end{definition}
Few small graphs are Seidel complement stable, for instance, 
$P_4$, $C_5$, and more trivially $K_{n}$ the clique on 
$n$ vertices and $S_{n}$ the stable on $n$ vertices.
\begin{lemma}[]
\label{lem:XF5-stable}
The graph $XF_{5}^{n}$ is Seidel complement stable.
\end{lemma}

\begin{proof}
\label{pr:XF5-stable}
$XF_{5}^{n}$ is a path of length $n$ dominated by two non-adjacent vertices 
$C$ and $D$. In addition to that,  a vertex $A$ is connected to $D$ and $1$, and 
a vertex $B$ is connected to $C$ and $n+1$. This graph is represented in Figure 
\ref{fig:XF5}.

~

The degree sequence for this graph for $n \geqslant 1$ is 
$[2;2;4 \times n; n+2; n+2]$. Except for $n=3$ the degree sequence 
allows us to ``identify'' the vertices. $A$ and $B$ are the vertices 
of degree $2$, $C$ and $D$ are the vertices of degree $n+1$ 
and the vertices of the path $[1,n+1]$ are the vertices of degree $4$.

~

Now let us formulate two easy observations. Since the graph 
presents of lot of symmetries, \ie $A$ is equivalent to $B$; 
$C$ is equivalent to $D$. It suffices to check that the 
graph obtained after a Seidel complement on the following vertices 
will preserve the desired properties. So the set of vertices 
to consider is $\{A,D,1,\ldots ,\lceil n+1\rceil \}$.

~

\noindent
Now two easy observations: $G$ denotes $XF_{5}^{n}$. 
$G \cong G *  D$. Since $D$ is connected to $\{A,1,\ldots,n+1\}$. 
After the Seidel complement it means that $C$ is now connected to 
only $B$ and $A$. And it also means that $B$ is connected to $C$, 
and since $B$ was only connected to $n+1$ in the original graph, 
it is now connected to $\{A,1,\ldots,n\}$. 
So now the path consists of the vertices $\{A,1,\ldots,n\}$, 
$B$ and $D$ dominate this path and $C$ and $n+1$ constitute the
extremities. The function $\varphi$ is given by this permutation.
$$\sigma =
\begin{pmatrix}
A & B & C & D & 1 &2 &\ldots & n+1 \\
1 & D & A & C & 2 & 3 &\ldots & ~B~ 
\end{pmatrix}
$$

\noindent
Let us show now that $G \cong G*A$.
By definition, the subgraph induced by $\{B,C,2,\ldots,n+1\}$ remains unchanged.
The vertex $1$ is now connected to $\{3,\ldots,n+1,B\}$, and is 
still connected to $A$ and $D$. Concerning $D$, it is now only connected 
to $B$ and $C$ in $G[\noN{A}]$. 
So the bijection $\varphi$ is given by the following permutation:
$$
\sigma =
\begin{pmatrix}
A &  B & C & ~D~ & 1 & 2 & \ldots & n+1 \\
A &  n & C & n+1  & D & B & \ldots & n-1 
\end{pmatrix}
$$
~\\
It is easy to see that $G\cong G*1$. 
The path is $3,4,\ldots,n+1,B,D,1,C$. 
The vertex $A$ is connected to $\{3,4,\ldots,n+1,B,D,1\}$
 and the vertex $2$ is connected to $\{4,\ldots,n+1,B,D,1,C\}$.
~\\
Let  us consider the case for the vertex $2$.
Actually $G \cong G * 2$ 
The path is $\{4,5,\ldots,n+1,B,D,2,C,A\}$.
and the vertex $1$ is connected to $\{4,5,\ldots,n+1,B,D,2,C\}$. 
And the vertex $3$ is connected to $\{5,\ldots,n+1,B,D,2,C,A\}$.

~
\\
Concerning the vertices on the path, let us consider the case 
of their vertex $k$ such that $k \in [3,n-1]$. It is 
clear that the graph $G[\{C,D,k-1,k,k+1\}]$ remains unchanged
as for the graph $G[V \setminus \{C,D,k-1,k,k+1\}]$. 
The vertex $C$ is now connected to $A,k-1,k \text{~and~} k+1$. So it is 
$4$. A similar thing happens for $D$. It is now connected to 
$B,k-1,k\text{~and~}k+1$. Concerning $k-1$ and $k+1$, 
$k-1$ is connected to every vertex except $k-2$ and $k+1$,
so its degree is $n+2$. And $k+1$ is connected to every vertex 
except $k-1$ and $k+2$. Concerning $A$ and $B$ their degrees
are now equal to $4$ (because of $k-1$ and $k+1$). And 
concerning the vertices $k-2$ and $k+2$ their degrees 
equal  $2$ because they are no longer connected to 
$C,D,k\pm 1$ but are now connected to $k\pm 1$ (\ie $k-1$ and $k+1$ swap roles).
\\
Now the extremities of the path are $k-2$ and $k+2$.
The path is of the form:
$k-2, \ldots, 1, A,C,k,D,B,n+1,n,\ldots,k+2$
\\
Consequently the graph $XF_{5}^{n}$ is Seidel complement stable.
\end{proof}

\begin{lemma}[]
\label{lem:HoleXF4}
The Seidel stable class of the hole $C_{n}$ consists of 
$C_{n}$, $XF_{4}^{n-6}$.
\end{lemma}
 Due to lack of space the proof is omitted, but in a few words, it relies 
on the ``regular'' 
structure of $XF_{5}^{n}$ and Lemma \ref{lem:XF5-stable}.
\begin{theorem}[Main Theorem]
\label{th:SeidelMinorPermutation}
A graph is a permutation graph \Iff it does not contain 
as finite graphs 
\FList and their complements  and as infinite families \IFList and 
their complements as Seidel minor.
\end{theorem}
\begin{proof}[Sketch of Proof]
This theorem relies on Gallai's result (\cf Theorem \ref{th:GallaiPermutation}). 
If $G$ is not a permutation graph then it contains one of the 
graphs listed in Theorem  \ref{th:GallaiPermutation} as an induced 
subgraph. Thanks to previous propositions 
\ref{prop:X3X2X36}-
\ref{prop:C6XF40} concerning the finite families, and 
propositions 
\ref{prop:Hole-SM-XF3XF4XF2}-
\ref{prop:C2n+1} concerning the infinite families.
We are able to reduce these induced subgraphs into 
a smaller set of graphs which are now forbidden Seidel minor.
It remains to prove that this list is minimal. 
Concerning the infinite families, Lemma \ref{lem:XF5-stable} proves that 
it is not possible to get rid of this families since it is 
Seidel stable. Concerning even holes (since odd holes 
are dismissed because they contain $XF_{5}^{2n-1}$ as Seidel minors) 
Lemma \ref{lem:HoleXF4} says that it is not possible 
to get rid of them. The same kind of argument holds for 
the finite graphs.
\end{proof}
\begin{corollary}[]
The class of permutation graphs is not well quasi ordered under Seidel minor relation.
\end{corollary}
\begin{proof}
$XF_{5}^{2n+3}$ constitutes an obstruction for permutation graphs. 
But for even values $XF_{5}^{2n}$ is a 
permutation graph. Furthermore, it is easy to check 
that for $k$ and $l$ two positives integers such that $k<l$, 
$XF_{5}^{2k}$ is not an induced subgraph of $XF_{5}^{2l}$.
Consequently the family $XF_{5}^{2n}$ is an infinite family 
of finite permutation graphs. Since $XF_{5}^{n}$  is Seidel stable 
by Lemma \ref{lem:XF5-stable}, these graphs are not 
comparable each other with the Seidel minor relation.
It is thus an infinite anti-chain for the Seidel minor relation 
and consequently permutation graphs are not well quasi ordered under Seidel minor 
relation.
\end{proof}
\section{Distance between Seidel complement equivalent graphs}
\label{sec:equivalence}
In this section we show that if two graphs $G$ and $H$ are Seidel complement 
equivalent, they are at distance at most $1$ from each other.

\subsection{General remarks}

\begin{lemma}
\label{lem:distance1}
 Let $G$ and $H$ be two graphs, $G$ and $H$ are Seidel complement equivalent 
if and only if they are at distance at most $1$.
\end{lemma}

\begin{proof}
If $G$ and $H$ are isomorphic they are Seidel complement equivalent by Definition 
\ref{def:SeidelEquivalent}. 
Let us assume that $G$ and $H$ are not isomorphic.
As an observation of the proof of Lemma \ref{prop:pivot} and Remark \ref{rem:pivotnn} 
we can notice that for any graph $G$ and any pair of vertices $v$ and $w$ of $G$ 
we have the following:
$$ G_1 = G *v *w  \cong G * w = G_2$$
This equality allows us to reduce, when given a sequence of Seidel complementation, 
by one. It is not longer reducible when the sequence is only of length $1$ and 
by hypothesis since $G$ and $H$ are not isomorphic it is the best we can do.
\end{proof}

\begin{corollary}[]
The number of graphs that are Seidel complement equivalent to a graph $G$ 
is at most $n+1$.
\end{corollary}

\begin{corollary}[]
Seidel complement equivalence is polynomially reducible to Graph Isomorphism 
problem.
\end{corollary}

\begin{lemma}[]
 Let $G$ and $H$ be two prime graphs \wrt modular decomposition, to decide 
if $G$ is isomorphic to $H$ is polynomially reducible to Seidel complement 
equivalence problem.
\end{lemma}

\begin{proof}
The reduction is as follows: $G'$ (resp. $H'$) is obtained from $G$ (resp. $H$)
 by adding a universal vertex $x$ (resp. $y$). As a consequence $G'$ has only 
one non-trivial module: $G$. The modular decomposition tree of $G'$ is 
composed of two internal nodes, the root is a series node with two children: 
the universal vertex and the prime node labeled by $G$.

\noindent 
{\bf Claim: } $G$ is isomorphic to $H$ iff $G'$ is Seidel complement equivalent to $H'$.
\begin{proof}
~\\
$(\Rightarrow)$ is obvious. \\
$(\Leftarrow)$  Any graph that is Seidel complement equivalent 
to $G'$ is actually isomorphic to $G'$. Since $G$ is a prime graph, 
its modular decomposition tree $T(G)$ is only a prime node labeled $G$. 
Since $G'$ is obtained from $G$ by adding a universal vertex to $G$, 
its modular decomposition tree is simply a series node as the root, his 
first child is the universal vertex and the second child is the modular decomposition 
tree of $G$. Thanks to Theorem \ref{th:MDtreeSeidel} and \ref{lem:distance1}, 
it is easy to realize that only one graph in the Seidel complement equivalence 
class of $G'$ possess a universal vertex. Actually  performing a Seidel complement 
at any vertex attached to the prime node of $T(G')$  will result in graph without 
a universal vertex since the root of modular decomposition tree obtained will be 
a prime node. And since all the graph that are Seidel equivalent to 
a given graph are at distance at most $1$, $G'$ and $H'$ are Seidel equivalent iff $G'$ and 
$H'$ are isomorphic iff $G$ and $H$ are isomorphic. 
\end{proof}
It is enough to conclude.
\end{proof}

\subsection{The case of cographs and permutation graphs}
In this section we show that for the class of cographs and permutation graphs. 
We can decide, in linear time for cographs and in quadratic time 
for permutation graphs if the graphs are given with their co-tree for cographs 
or with their intersection model for permutation graphs.

\begin{lemma}[]
\label{lem:SeidelCographEquivLinear}
To decide if two cographs $G$ and $H$ are Seidel complement equivalent 
can be computed in linear time $O(n)$.
\end{lemma}
\begin{proof}
Let us consider the co-trees $T(G)$ and $T(H)$. We modify 
$T(G)$ and $T(H)$ as follows: 
Let $T'(G)$ be the co-tree of $G$ on which we add a dummy vertex 
attached to the root of $T(G)$. We proceed in a similar manner 
for $T'(H)$. 

$G$ and $H$ are Seidel complement equivalent \Iff $T'(G)$ and $T'(H)$ 
are isomorphic.
\\
$\Rightarrow$ This direction is easy, since according to  
Remark  \ref{rem:exchange}, Theorem \ref{th:SeidelCograph} 
 and Lemma \ref{lem:distance1}, if 
$G$ and $H$ are Seidel complement equivalent then $T'(G)$ and 
$T'(H)$ are isomorphic.
\\
$\Leftarrow$
Let us assume now that $T'(G)$ and $T'(H)$ are isomorphic and 
let $\varphi : V(T'(G)) \mapsto V(T'(H))$ be the mapping function.
The isomorphism considered here is the labeled isomorphism, \ie 
labels of the internal nodes, 0 or 1, are preserved. 

Using the result of Theorem \ref{lem:distance1}, we know 
that cographs are at distance at most 1. It is thus sufficient 
to find the actual vertex to transform one co-tree into another.

Let us call the dummy vertices added to turn $T(G)$ (\resp $T(H)$) 
into $T'(G)$ (\resp $T'(H)$) $du_{G}$ and $du_{H}$. 
Now  since we want to transform $T(H)$ into $T(G)$ it suffices 
to pick a vertex $f$ in $T(H)$ such that it is the image 
by $\varphi$ of $du_{G}$ \ie $f = \varphi(du_{G})$. Once we have obtained 
this vertex in $T(G)$ it is sufficient to proceed to a Seidel complement 
on $f$, $H*f$, so now $P(f)$ is the root of $T(H*f)$ as requested since 
$f$ was an image of $du_{G}$ and $f$ is now attached to the former root 
$R(H)$. Consequently we have shown that when $T'(G)$ and $T'(H)$ are 
isomorphic we can find a vertex permitting us to transform $T(H)$ into $T(G)$
and hence proving that they are Seidel complement equivalent.

This procedure can be achieved in linear time, since deciding if 
two given trees are isomorphic is well known to be linear \cite{AhoHU74}, 
and finding the actual vertex and performing the Seidel complementation 
is done in constant time.
\end{proof}
\begin{lemma}
 Let $G$ and $H$ be two permutation graphs given with their representation 
then we can decide in $O(n^2)$ if $G$ is Seidel complement equivalent to $H$.
\end{lemma}
\begin{proof}
From Lemma \ref{lem:distance1} we know that every Seidel complement 
equivalent graphs are at distance at most one from each other. It suffices, 
\Wlog to apply on $H$ every possible Seidel complementation, and check 
if one the graph obtained is isomorphic to $G$. We can decide 
in $O(n)$ time if two permutation diagrams (and hence the corresponding 
permutation graphs). Finally, in the worst case we must 
try on all the vertices of $H$ if $G \cong H * v$. Seidel complementation 
on the permutation diagram can be performed in constant time, 
the final complexity is $O(n^2)$.
\end{proof}
\section{Tournaments}
\label{sec:tournaments}
We present in this section a notion of Seidel complementation applied to 
tournaments.

Let $T=(V,A)$ be a tournament. And let $v$ a vertex of $T$. 

\begin{definition}
\label{def:SeidelCtournament}
The Seidel complementation on a tournament $T$ applied on a vertex $v$ is defined as follows:
\begin{itemize}
 \item reverse the direction of all the arcs lying between $N^{+}(v)$ and $N^{-}(v)$.
 \item reverse all the arcs incident to $v$, \ie $N^+(v)$ becomes $N^-(v)$ and conversely.
\end{itemize}
\end{definition}

\begin{lemma}
Let $T$ be a tournament and let $v$ be a vertex of $T$.
$T$ is prime \wrt modular decomposition iff $T*v$ is prime.
\end{lemma}

\begin{proof}
The proof is almost the same as the  proof of Theorem \ref{th:SeidelMDPrime}, 
it suffices to replace the edges and non edges by the arcs.
\end{proof}

The modular decomposition tree of a tournament is modified in the same way 
as described in Theorem \ref{th:MDtreeSeidel} for undirected graphs. It suffices 
to use Definition \ref{def:SeidelCtournament} instead of \ref{def:SeidelSwitchN}.

\section{Conclusion and Perspectives}

We have shown that the new paradigm of Seidel minor provides 
a nice and compact characterization of permutation graphs. 

A lot of questions remain open. 
\\
A natural question lies in the fact that Theorem \ref{th:SeidelMinorPermutation} 
is obtained using Gallai's result on forbidden induced subgraphs. Is it 
possible to give a direct proof of  Theorem \ref{th:SeidelMinorPermutation} without using 
Gallai's result?
\\
Another direction concerns graph decomposition. Oum
\cite{Oum05} has shown that local complementation preserves rank-width. Is 
there a graph decomposition that is preserved by Seidel complementation?
\\
Finally, it could be interesting to generalize the Seidel complement operator 
to directed graphs, and possibly to hypergraphs.

We hope that this Seidel minor will be relevant in the future as a tool 
to study graph decomposition and to provide similar characterizations, as 
the one presented for permutation graphs,  to other graph classes.

\paragraph{\bf Acknowledgement}
The author is grateful to M. Bouvel, B. Courcelle, D. Corneil, M.C. Golumbic, 
M. Habib, M. Kant\'e, F. de Montgolfier and M. Rao for fruitful discussions and for pointing out 
relevant references. 

\bibliography{Seidel}
\bibliographystyle{amsplain}
\end{document}